\documentclass{article}
\usepackage[utf8]{inputenc}
\usepackage[margin=1in]{geometry}

\usepackage{microtype}
\usepackage{graphicx}
\usepackage{subfigure}
\usepackage{booktabs}
\usepackage{natbib}
\setcitestyle{open={(},close={)}}
\usepackage{amsmath}%
\usepackage{color}%
\usepackage{xspace}%
\usepackage{hyperref}%
\hypersetup{%
   breaklinks,%
   colorlinks=true,%
   linkcolor=[rgb]{0.45,0.0,0.0},%
   citecolor=[rgb]{0,0,0.7}
}

\numberwithin{table}{section}%
\numberwithin{equation}{section}%

\usepackage{amsmath,amsthm,amsfonts,amssymb,mathdots,array,mathrsfs,bm,bbm,stmaryrd,graphicx,subfigure,xcolor}
\usepackage{algpseudocode,algorithm,algorithmicx}
\usepackage[algo2e,boxed,noend]{algorithm2e}
\usepackage{breakcites}

\usepackage[T1]{fontenc}
\usepackage{enumerate}
\usepackage{inputenc}

\usepackage{graphicx} 
\usepackage{subfigure}
\usepackage{booktabs,balance}
\usepackage{rotating}
\usepackage{boldline}
\usepackage{makecell}
\usepackage{multirow}
\usepackage{balance}

\usepackage{tikz}

\newcommand{\eqdef}{\mathrel{\mathop:}=}

\newtheorem{theorem}{Theorem}[section]

\newtheorem{lemma}[theorem]{Lemma}

\newtheorem{definition}{Definition}[section]

\newcommand{\bsmat}{\begin{bmatrix} }
\newcommand{\esmat}{\end{bmatrix} }

\newcommand{\HLinkShort}[2]{\hyperref[#2]{#1\ref*{#2}}}
\newcommand{\HLink}[2]{\hyperref[#2]{#1~\ref*{#2}}}
\newcommand{\HLinkPage}[2]{\hyperref[#2]{#1~\ref*{#2}%
		$_\text{p\pageref{#2}}$}}
\newcommand{\HLinkPageOnly}[1]{\hyperref[#1]{Page~\refpage*{#1}%
		$_\text{p\pageref{#1}}$}}

\newcommand{\HLinkSuffix}[3]{\hyperref[#2]{#1\ref*{#2}{#3}}}
\newcommand{\HLinkPageSuffix}[3]{\hyperref[#2]{#1\ref*{#2}%
		#3$_\text{p\pageref{#2}}$}}

\newcommand{\deflab}[1]{\label{def:#1}}
\newcommand{\defref}[1]{\HLink{Definition}{def:#1}}%

\newcommand{\algolab}[1]{\label{algo:#1}}
\newcommand{\algoref}[1]{\HLink{Algorithm}{algo:#1}}%

\newcommand{\figlab}[1]{\label{fig:#1}}
\newcommand{\figref}[1]{\HLink{Figure}{fig:#1}}

\providecommand{\lemlab}[1]{\label{xlemma:#1}}
\renewcommand{\lemlab}[1]{\label{xlemma:#1}}
\newcommand{\lemref}[1]{\HLink{Lemma}{xlemma:#1}}%

\newcommand{\thmlab}[1]{{\label{theo:#1}}}
\newcommand{\thmref}[1]{\HLink{Theorem}{theo:#1}}

\newcommand{\seclab}[1]{\label{sec:#1}}
\newcommand{\secref}[1]{\HLink{Section}{sec:#1}}

\providecommand{\eqlab}[1]{}%
\renewcommand{\eqlab}[1]{\label{equation:#1}}

\DeclareMathOperator*{\argmin}{argmin}

\begin{document}

\title{\bf Breaking the Linear Error Barrier in Differentially Private Graph Distance Release}

\author{\vspace{0.3in}\\\textbf{Chenglin Fan, Ping Li, Xiaoyun Li} \\\\
Cognitive Computing Lab\\
Baidu Research\\
10900 NE 8th St. Bellevue, WA 98004, USA\\
  \texttt{\{chenglinfan2020,\ pingli98, lixiaoyun996\}@gmail.com}
}
\date{\vspace{0.1in}}
\maketitle

\begin{abstract}
\vspace{0.2in}
\noindent\footnote{Initially submitted in 2021.}Releasing all pairwise shortest path (APSP) distances between vertices on general graphs under weight Differential Privacy (DP) is known as a challenging task. In the previous attempt of~\cite{DBLP:conf/pods/Sealfon16}, by adding Laplace noise to each edge weight or to each output distance, to achieve DP with some fixed budget, with high probability the maximal absolute error among all published pairwise distances is roughly $O(n)$ where $n$ is the number of nodes. It was shown that this error could be reduced for some special graphs, which, however, is hard for general graphs. Therefore, whether the approximation error can be reduced to sublinear in $n$ is posted as an interesting open problem~\citep{DBLP:conf/pods/Sealfon16,DPorg-open-problem-all-pairs}.

\vspace{0.15in}
\noindent In this paper, we break the linear barrier on the distance approximation error of previous result, by proposing an algorithm that releases a constructed synthetic graph privately. Computing all pairwise distances on the constructed graph only introduces $\tilde  O(n^{1/2})$ error in answering all pairwise shortest path distances for fixed privacy parameter. Our method is based on a novel graph diameter (link length) augmentation via constructing ``shortcuts'' for the paths.
By adding a set of shortcut edges to the original graph, we show that any node pair has a shortest path with link length $\tilde O(n^{1/2})$.
Then by adding noises with some positive mean to the edge weights, we show that the new graph is differentially private and can be published to answer all pairwise shortest path distances with $\tilde O(n^{1/2})$ approximation error using standard APSP computation. Numerical examples are also provided.

\vspace{0.15in}
\noindent Additionally, we also consider the graph with small feedback  vertex set  number. A feedback vertex set (FVS) of a graph is a set of vertices whose removal leaves a graph without cycles, and the feedback vertex set number of a graph, $k$, is the size of a smallest feedback vertex set. We propose a DP algorithm with error rate $\tilde O(k)$, which improves the error of general graphs provided $k=o(n^{1/2})$, i.e., the graph has small feedback vertex set number.

\end{abstract}

\newpage

\section{Introduction}

In recent years, there has been a growing interest in private data analysis,  from academic research to industry practice. Typically, many industrial machine learning applications consist of two steps: 1) collecting data from users; 2) training models using the collected user data. The key question is: throughout this process, how can we protect the privacy of each individual, without leaking sensitive data of each user? In other words, how can we prevent an adversarial attacker from inferring any user's data, given the public information that can be accessed? This has been the overarching question in a line of research that studies private algorithms under various settings. The  concept of differential privacy (DP)~\citep{DBLP:conf/tcc/DworkMNS06,DBLP:conf/pods/BlumDMN05,DBLP:conf/tcc/ChawlaDMSW05,DBLP:conf/icalp/Dwork06} has been a popular approach for rigorously defining and resolving the problem of keeping useful information for model learning, while protecting privacy for each individual. In the traditional setting, databases $D$ and $D'$ are collections of data records and are considered to be neighboring if they are identical except for a single record (individual information). DP requires that the output of running a randomized algorithm on $D$ and $D'$ should have very close probability distributions. There are many types of databases, and in this work, we will specifically focus on the differential privacy of graphs.

\vspace{0.15in}
\noindent\textbf{Weight private graphs.}
In general, there are three types of  private graph models, regarding the nodes, the edges and the edge weights, respectively. The node private model requires DP for two adjacent graphs differing in one node. In the edge private model, two neighboring graphs are defined such that they only differ by one edge, while sharing the same set of nodes.

In this paper, we focus on the weight private graph model, which was first proposed by~\citet{DBLP:conf/pods/Sealfon16}.
In the weight private graph problem, the topology of the graph is public, which means that the adjacent graphs (databases) when considering DP have same nodes and edges, while this is not the case in other two DP models. Specifically, in the weight private model we consider two graphs being adjacent if the difference between the sums of their edge weights is no more than one unit. One example application of this problem is where one tries to release the transportation volume between several places of interest with privacy constraints, where the roads are regarded as edges. In some sense, it might be hard to keep the graph structure private in practice, since one may easily obtain the map by modern tools like Google Earth.  However, in some cases, information like traffic flows could also be sensitive and needs to be kept private. Releasing such private information is well suited for the setting of weight private graph model.

\subsection{Open Problem, and Previous Results} \seclab{sec:open problem}

In this paper, we study the problem of releasing all pairwise shortest-path (APSP) distances with weight privacy. Two standard strategies in DP are: 1) adding Laplace noise to each edge, 2) adding Laplace noise to the output distance matrix. As discussed in~\cite{DBLP:conf/pods/Sealfon16}, the error of these two strategies is roughly $O(\epsilon^{-1}n\log n)$ by adding noise to each edge and $O(\epsilon^{-1}n\sqrt{\log(1/\delta)})$ using DP composition theorems, respectively, where $n$ is the number of nodes in the graph and $(\epsilon,\delta)$ is the DP parameter. Here, the approximation error is measured by the largest absolute difference among all released pairwise distances and the ground truth. As we see, both of these errors are roughly $O(n)$, i.e., linear in $n$ when $\epsilon$ and $\delta$ are fixed. Whether we can achieve differential privacy with error sublinear in $n$ is still an open problem, as recognized by an online repository (DifferentialPrivacy.org\footnote{\url{https://differentialprivacy.org/open-problem-all-pairs/}}) focusing on the research of DP, quoted \textit{``... Is this linear dependence on $n$ inherent, or is it possible to release all-pairs distances with error sublinear in n ?''}

While this linear barrier seems to hold for general graphs, we can derive improved results on several special types of graph. Firstly, for a graph with weights bounded in $[0,M]$,~\cite{DBLP:conf/pods/Sealfon16} picked a subset $Z\in \mathbb{V}$ of vertices as the ``$k$-covering set'' to approximate the original graph.  Each vertex $u$ can map to its closest  vertex $z_u$ in the covering set. For any node pair $u,v$, their distance is approximated by the distance between $z_u$ and $z_v$ plus $O(kM)$ additional error. Then, the author proposed to use $O(|Z|^2)$ pairs of distances to approximate all $O(n^2)$ pairwise distances, which leads to $O(\sqrt{n M})$ approximation error for fixed privacy parameters eventually. Secondly, for grid graph with arbitrary positive weights,~\cite{fan2022distances} proposed to select an intermediate vertex set and divide the shortest path between any node pair into at most three parts depending on its traverse of the set. By connecting the nodes in the intermediate set and applying the standard Laplace mechanism, the authors constructed a DP algorithm with $\tilde O(n^{3/4})$ approximation error.
Thirdly, for trees, ~\cite{DBLP:conf/pods/Sealfon16} gave a recursive algorithm to release all-pairs distances with error
$O(\log^{2.5}
n)$, which was improved to $O(\log^{1.5} n\log^{1.5} h)$ by~\cite{fan2022distances}, where $h$ is the depth of tree and can be as small as $O(\log n)$. The general idea was to build a collection $\mathbb{T}$ of subtrees of $\log n$ levels, where the subtrees in each level are disjoint to each other such that each edge appears in at most $O(\log n)$ trees, allowing us to only add $O(\log n)$ units of Laplace noise to each path in $\mathbb T$. This leads to $poly\log n/\epsilon$ approximation error on trees.

\subsection{Our Results and Techniques}

In this paper, we propose an algorithm that is able to surpass the linear $O(n)$ error of differentially privately releasing all pairwise distances on general weighted graphs. We make use of the improved and extended techniques from~\citet{fan2022distances}
to general graphs to obtain  $\tilde O(n^{1/2})$ approximation error. Moreover, when the graph exhibits some special property having small feedback vertex set number, we develop a new algorithm with further improved error bound.

\vspace{0.1in}
\noindent\textbf{Breaking the linear error barrier for general graphs.} Let $G(\mathbb V,\mathbb E,w)$ denote an undirected graph, where $\mathbb V$ is the set of nodes with $|\mathbb V|=n$, $\mathbb E$ is the set of all edges, and $w$ is the set of corresponding weights. We use $w(e)$ to denote the weight of a specific edge $e\in\mathbb E$. The key challenge in this problem is that, we need to bound the error of all shortest paths instead of just one. Since the shortest paths could overlap with each other, the dependency among the shared edges (and the noises) hinders us from applying standard concentration results. To achieve sublinear error in DP all pairwise distance release problem, our idea, intuitively, is to find ``shortcuts'' that will reduce the number of edges in long shortest paths (consisting of many edges). The algorithm proceeds as follows. Denote $d(u,v)$ as the true shortest path distance between node $u$ and $v$ w.r.t. graph $G(\mathbb V,\mathbb E,w)$. First, we randomly sample $n^{1/2}$ vertices from $\mathbb V$ to form a subset of vertices $\mathbb V_1\subset \mathbb V$. We then
create an edge set  $\mathbb{E}_1=\{(u,v): u\in \mathbb V_1,v\in \mathbb V_1,u\neq v\}$ and assign each edge the weight $w_1(u,v)=d(u,v)$. Denote $\mathbb E_0=: \{(u,v)\in \mathbb E: u\notin \mathbb V_1 \text{ or } v\notin \mathbb V_1\}$ as the edge set between nodes that are not both in $\mathbb V_1$. Next, we add noise to the edges in $\mathbb E_0$ and $\mathbb E_1$, respectively. For $e\in\mathbb E_0$, we add Laplace noise following $Lap(\mu_0,\sigma_0)$ where $\sigma_0=O(1/\epsilon),\mu_0=\sigma \log (n/\gamma)$; for $e\in\mathbb E_1$, the noise is from $Lap(\mu_1,\sigma_1)$ with $\sigma_1:= O(n^{1/2}\sqrt{\log (1/\delta)})/\epsilon)$ and $\mu_1=\sigma_1 \log (n/\gamma)$. Note that, the mean of the Laplace noise is non-zero in our approach. We will show that this novel design guarantees that the APSP distances computed directly from the synthetic and noisy graph $G'$ are lower bounded by the true distances for all node pairs (w.h.p.), and the approximation error is at most $\tilde O(n^{1/2})$ (\thmref{main}).

\vspace{0.1in}
\noindent\textbf{Graph with small feedback vertex set number.}
A feedback vertex set (FVS) of a graph, also called a loop cutset~\citep{DBLP:journals/jacm/Freuder82}, is a set of vertices whose removal leaves a graph without cycles, and the feedback vertex set number of a graph is the size of a smallest feedback vertex set.
The feedback vertex set problem is NP-hard~\citep{DBLP:conf/stoc/Lewis78} and the best known approximation algorithm on undirected graphs is by a factor of two~\citep{DBLP:journals/ai/BeckerG96}.
Finding the size of a minimum feedback vertex set can be solved in $O(1.7347^n)$ time~\citep{DBLP:conf/swat/CaoCL10}, where $n$ is the number of vertices in the undirected  graph. Beyond the results derived for general weighted graphs, we also consider the problem for graphs with small feedback vertex set. We design a new algorithm that releases the distances privately, with $\tilde O(k)$ error where $k$ is the feedback vertex set number. This implies improved error bounds when $k=o(n^{1/2})$ compared with the result for general weighted graphs.

\subsection {More Related Work}

The topic of private graphs has attracted substantial research interests in the recent years~\citep{DBLP:conf/icdm/HayLMJ09,DBLP:conf/pods/RastogiHMS09,DBLP:conf/soda/GuptaLMRT10,DBLP:journals/pvldb/KarwaRSY11,DBLP:conf/tcc/GuptaRU12,DBLP:conf/innovations/BlockiBDS13,DBLP:conf/tcc/KasiviswanathanNRS13,DBLP:conf/focs/BunNSV15,DBLP:conf/pods/Sealfon16,  DBLP:conf/nips/UllmanS19,DBLP:conf/focs/BorgsCSZ18,DBLP:conf/nips/AroraU19,fan2022distances}. Our work focuses on weight private graphs and generalizes the previous work~\citep{fan2022distances} which focused on trees and grid graphs.  We are also aware of two concurrent works~\citep{Article:Ghazi_2022,chen2022all}.

Below, we briefly introduce more related work on three types of privacy on graphs.
\begin{itemize}

\item {\bf Edge Privacy.}~\cite{DBLP:conf/icdm/HayLMJ09} constructed a differentially edge-private algorithm for releasing the degree distribution of a graph.~\cite{DBLP:conf/tcc/GuptaRU12} showed how to answer cut queries in a private edge model.~\cite{DBLP:conf/focs/BlockiBDS12} improved the error for small cuts.~\cite{DBLP:conf/soda/GuptaLMRT10} showed
how to privately release a cut close to the optimal error size.~\cite{DBLP:conf/nips/AroraU19} studied the private sparsification of graphs, which was exemplified by a proposed graph meta-algorithm for privately answering cut-queries with improved accuracy.

\item \textbf{Node Privacy.} Node differential privacy (node-DP) requires the algorithm to hide the presence or absence of a single node and the (arbitrary) set of edges incident to that node~\citep{DBLP:conf/nips/UllmanS19}. However, node-DP is often difficult to achieve without compromising accuracy, because even very simple graph statistics can be highly sensitive to adding or removing a single node.~\cite{DBLP:conf/tcc/KasiviswanathanNRS13} studied the problem of  releasing the entire degree distribution of a graph with node privacy. The work of~\citet{DBLP:conf/innovations/BlockiBDS13} also considered node-level differential private algorithms for analyzing sparse graphs.~\cite{DBLP:conf/focs/BorgsCSZ18} gave a simple, computationally efficient, and node-DP algorithm for estimating the parameter of an Erd\H{o}s-R\'{e}nyi graph.

\item \textbf{Weight Privacy.} As we stated before, the weight private graph model was first proposed by Sealfon~\citep{DBLP:conf/pods/Sealfon16}, which has been summarized in \secref{sec:open problem}. The problem of releasing approximate distances on the path graph (i.e., the graph is composed of only one single path) is equivalent to approximating all threshold functions on a totally ordered
set, whose error bound was obtained in~\citet{DBLP:conf/focs/BunNSV15}.
in~\citet{DBLP:conf/infocom/GhoshDS020}, the authors considered the problem of privately reporting counts of events recorded by devices in different regions of the plane and used a novel hierarchical planar separator to answer queries over arbitrary planar graphs. However, that planar graph is non-private and those techniques can not be used here directly.
When considering our distance release problem specific on path graphs, it is also equivalent to answering all range queries on a histogram under differential privacy, which has been studied in literature, e.g., the matrix mechanism~\citep{DBLP:conf/pods/LiHRMM10}.

\end{itemize}

\section{Background}\label{sec:prelims}

Consider a graph $G=(\mathbb{V},\mathbb{E},w)$, with $w$ the collection of all weights of $\mathbb E$. We use $w(e)$ to denote the weight of an edge $e\in \mathbb{E}$.
Denote $n=|\mathbb{V}|$, $m=|\mathbb{E}|$, and we assume $G$ is connected such that $n-1\leq m$.

\subsection{Differential Privacy (DP)}

We first define the notion of neighboring graphs. In weight private model, since the nodes and edges of the graph $G$ are unchanged, we can simply use edge weights to represent the graph.

\begin{definition}[Neighboring]
Graph $G=(\mathbb{V},\mathbb{E},w)$ and $G'=(\mathbb{V},\mathbb{E},w')$ are called neighboring, noted as $G \sim G'$, if
$$ ||w-w'||_1\eqdef \sum_{e\in \mathbb{E}} |w(e)-w'(e)| \leq 1. $$
\end{definition}

The Differential Privacy (DP) introduced by~\citet{DBLP:conf/icalp/Dwork06} is defined below adapted to our problem.

\begin{definition}[Differential Privacy~\citep{DBLP:conf/icalp/Dwork06}]
If for any two neighboring graphs $G=(\mathbb{V},\mathbb{E},w)$ and $G'=(\mathbb{V},\mathbb{E},w')$, a randomized algorithm $\mathbbm{A}$, and a set of outcomes $O\subset Range(\mathbbm{A})$, it holds that
$$ Pr[\mathbbm{A}(G)\in O] \leq e^{\epsilon} Pr[\mathbbm{A}(G')\in O]+\delta,$$
we say algorithm $\mathbbm{A}$ is $(\epsilon,\delta)$-differentially private.
\end{definition}

If $\delta=0$, we say that the algorithm is $\epsilon$-DP. The parameter $\delta$ is usually interpreted as the probability allowed for bad cases where $\epsilon$-DP is violated. Intuitively, differential privacy requires that after changing the database by a little (the total weights in our case), the output should not be too different from that of the original database. Differential privacy may be achieved through the introduction of noise to the output. To attain general $(\epsilon,\delta)$-DP, we may add Gaussian noise to the output. To achieve $\epsilon$-DP, the noise added typically comes from the Laplace distribution. Smaller $\epsilon$ and $\delta$ indicate stronger privacy, which, however, usually sacrifices utility. Thus, one of the central topics in the differential privacy literature is to reduce the scale of noise added, while satisfying the privacy constraint.

In this work, we focus on the popular approach to achieve DP by adding Laplace noises. The Laplace distribution with parameter $b$ has density function $f(x)=\frac{1}{2b}\exp(-|x|/b)$.

\begin{lemma}[\cite{DBLP:conf/icalp/Dwork06}] \lemlab{lem:basic}
For a function $f:\mathcal G\rightarrow \mathbb{R}$ with $\mathcal G$ the input space of graphs, define the sensitivity
$$\triangle_f=\max_{G \sim G'} || f(G)-f(G')||_1,$$
where $G,G'$ are two neighboring graphs. Let $X$ be a random noise drawn from $Lap(0,\triangle_f/\epsilon)$. The Laplace mechanism outputs
$$ M_{f,\epsilon}(G)  = f(G)+X,$$
which achieves $\epsilon$-differential privacy.
\end{lemma}

One important and attractive property of DP is that, different DP algorithms can be easily combined together, also with strict DP guarantee.

\begin{lemma}[Advanced Composition Theorem~\citep{DBLP:conf/focs/DworkRV10}] \lemlab{lem:ACT}
For any $\epsilon, \delta, \delta' \geq  0$, the adaptive composition of $k$ times
$(\epsilon, \delta)$-differentially private mechanisms is $(\epsilon', k\delta + \delta'
)$-differentially private for
$$ \epsilon’ = \sqrt{2k\log(1/\delta’)} \cdot \epsilon + k \cdot \epsilon(e^\epsilon-1),$$
which is $O(\sqrt{k\log (1/\delta')}\cdot \epsilon)$ when $k \leq 1/\epsilon^2$. In particular, if $\epsilon' \in (0,1),\delta,\delta'>0$, the composition of $k$ times $(\epsilon,0)$-differentially private mechanism is $(\epsilon',\delta')$-differentially private for
$$ \epsilon=\epsilon'/(\sqrt{8k \log (1/\delta')}). $$

\end{lemma}


\subsection{Basic Probability Fact}

We will now introduce a few basic tools which will be used in the error analysis.
A number of differential privacy techniques incorporate noise sampled according to the Laplace
distribution.

\vspace{0.15in}

\begin{lemma}\lemlab{lem:single}
Consider $n$ i.i.d. random variables $Z_1,Z_2,...,Z_n$ from $Lap(b)$. With probability at least $1-\gamma$, $\forall 0<\gamma<1$, all $n$ Laplace random variables have magnitude bounded by $b \log(n/\gamma)$.
\end{lemma}

\begin{proof}
For each $Z_i$, $P(|Z_i|\leq t)=1-\exp(-\frac{t}{b})$. Thus,
\begin{align*}
    P[\max_{i=1,...,n} |Z_i|\leq t]= (1-\exp(-\frac{t}{b}))^n.
\end{align*}
Let $t=b\log(n/\gamma)$. With $n\geq 1$, the probability becomes $(1-\frac{\gamma}{n})^n\geq 1-\gamma$ as claimed.
\end{proof}

\section{Private Synthetic Graph Release for APSP Distance}

We formally define the problem of interest as below.
\begin{definition}  \deflab{distance}
(Approximate Distances  Release)
Given an graph $G(\mathbb{V},\mathbb{E},w)$, our task is to release all pairwise shortest path distances privately. Let $\hat d(\cdot,\cdot)$ be the output (approximate) distance function. Our object is to  minimize the maximal absolute error over all pairs, namely, $\max \{|\hat d(u,v)-d(u,v)|: u,v\in \mathbb V \}$, with $\hat d(u,v)$, $\forall u,v\in \mathbb V$ being DP to weight $w$. We call $|\hat d(u,v)-d(u,v)|$ the additive error of node pair $(u,v)$.
\end{definition}

Several basic notions will be used in our analysis.

\begin{definition}[Path]
A path $\tilde P_{v_1,v_L}$ between $v_1, v_L\in\mathbb V$ is defined as $\tilde P_{v_1,v_L}=\{(v_i,v_{i+1}), 1 \leq i\leq L-1\}$, the collection of edges between connected node sequence $v_1,v_2,...,v_k$, with some $L\leq n$.  A segment of path $\tilde P$ is defined as a \textbf{consecutive} sub-path, $\{(v_i,v_i+1),s\leq i \leq t-1\}$ for some $s,t$.
\end{definition}

\begin{definition}[Shortest path] \deflab{def:shortest path}
Let $S_{u,v}$ be the set of all paths between $u, v \in \mathbb{V}$. For a path $P \in S_{u,v}$ and weights $w$, denote $d(P,\mathcal W)=\sum_{e\in P}w(e)$. The shortest path of $u,v\in\mathbb V$ w.r.t. weights $w$ is
$$P_{u,v}=\argmin_{P \in S_{u,v}} d(P,w),$$
and the shortest path distance is $d(P_{u,v},w)$.
\end{definition}

\begin{definition}[Canonical  shortest path] \deflab{def:canonical}
For a given graph $G(\mathbb V,\mathbb E,w)$, let $\mathbb V_1\subseteq \mathbb V$ be a subset of $\mathbb V$. Let $P_{x,y}=\{x,v_1,v_2,...,v_L,y\}$ be the shortest path between $x,y\in\mathbb V$. Then a canonical shortest path $P^{\mathbb V_1}_{x,y}$ is defined by either of the following:
\begin{enumerate}
    \item $P^{\mathbb V_1}_{x,y}\equiv P_{x,y}$, if $P_{x,y}$ contains at most one vertex in $\mathbb V_1$;

    \item $P^{\mathbb V_1}_{x,y}=\{x,v_1,...,p,q,...,v_L,y\}$, if $p,q$ are the closest nodes in $P_{x,y}$ to $x,y$ respectively, and $p,q\in\mathbb V_1$.
\end{enumerate}

\end{definition}

Intuitively, the canonical shortest path finds a shortcut by directly connecting two nodes in a set $\mathbb V_1$. This definition will be the key in our algorithm and construction. One important fact is that, if we connect each pair of nodes in $\mathbb V_1$ and assign the true pairwise distance between them as the edge weight (as in our main algorithm), then $d(P_{u,v},w)\equiv d(P^{\mathbb V_1}_{u,v},w)$ by definition.

\subsection{Challenges and the New Algorithm}
	
First, we revisit the prior approach to release  all pairwise distances in the private weight model and the challenges. Since neighboring (total) weights differ by at most one unit in $l_1$ norm, the distance between any two nodes also changes by at most one unit.  Releasing a single path can be done by computing the accurate shortest path distance between pairs of inputs and adding Laplace noise proportional to $1/\epsilon$, which is a trivial task in DP. However, releasing all distances privately is much more challenging, since the Laplace mechanism requires $Lap(n^2/\epsilon)$ noise (because of the $n^2$ queries), resulting in the $O(n)$ error eventually. As introduced in \secref{sec:open problem}, in prior literature, there are two ways (e.g.,~\cite{DBLP:conf/pods/Sealfon16}) to achieve this error level, either by adding noise to each edge or to the output distances. In this paper, we will focus on the first strategy.

The simple approach is as follows: 1) add $Lap(1/\epsilon)$ noise to each edge, i.e., $w'=w+Lap(1/\epsilon)$, to get graph $G'=(\mathbb V,\mathbb E,w')$; 2) report all pairwise distances on $G'$ (note that, the shortest paths on $G'$ found might be different from the true shortest paths in $G$).
By the Laplace mechanism~\citep{DBLP:conf/icalp/Dwork06}, all the weights in graph $G$ become differentially private. By the post-processing property of DP, all the output pairwise distances are also DP. Note that there are $O(n^2)$ pairs of vertices, so the number of edges is bounded by $O(n^2)$. By \lemref{lem:single}, with probability $1-\gamma$, all $O(n^2)$ Laplace random variables will have magnitude bounded by $(1/\epsilon) \log(n^2/\gamma)$, so the length of every path in the released synthetic graph is within $n \log(n/\gamma)/\epsilon$ additive error, thus roughly $O(n)$.

Statistically, the problem can be described informally and approximately as: given a set of $n^2$ i.i.d. Laplace random variables, we want to bound the sum of  the $n$ variables among $n^2$ size-$n$ subsets simultaneously. To our best knowledge, this problem has no solution sublinear in $n$ in literature, and a straightforward additive bound exactly leads to the previous $O(n)$ error as mentioned above. In this section, we propose an algorithm that leverages the concept of canonical shortest paths (Definition~\defref{def:canonical}), which finally leads to $\tilde O(n^{1/2})$ additive approximation error.

\begin{algorithm2e}[t!]
\SetKwInput{Input}{Input}
\SetKwInOut{Output}{Output}
\Input{General graph $G=(\mathbb{V},\mathbb{E},w)$, private  parameter $\epsilon,\delta,\gamma$.}
\DontPrintSemicolon
$\epsilon'=\epsilon/2$.\;
	Sample $(n^{1/2})$ vertices from $\mathbb V$ uniformly and add them to $\mathbb V_1$.\;
	
Create an edge set  $\mathbb{E}_1:=\{(u,v): u\in \mathbb V_1,v\in \mathbb V_1,u\neq v\}$.\;
For each $e\in \mathbb{E}_1 $,  let $APSP(e)$   be the exact shortest path distance in $G$ between $(u,v)$ where $e=(u,v)$.\;
\For {each edge $e\in \mathbb{E}_1$ }
{
Let 	$\sigma_1:= (2\sqrt{2} n^{1/2}\sqrt{\log (1/\delta)})/\epsilon'$ and $\mu_1:= \sigma_1 \log (n/\gamma)$.\;
	Compute $X_e:= Lap(\mu_1,\sigma_1)$.\;
	$w'(e):= APSP(e)+X_e$.\;
}

$\mathbb{E}_0:= \mathbb E\setminus \mathbb{E}_1$\;
\For {each edge $e\in  \mathbb{E}_0$}
{
	Let $\sigma_0:= 1/\epsilon',\mu_0:= \sigma_0 \log (n^2/\gamma)$.\;
	Compute $X_e:= Lap(\mu_0,\sigma_0)$.\;
	$w'(e):= w(e)+X_e$.\;
}
Compute all pairwise distances in graph
$G'=(\mathbb{V},\mathbb{E}'=\mathbb{E}_0 \cup \mathbb{E}_1 ,w')$.\;
\Output{All pairwise distances of $G'$.}
\caption{Private all pairwise shortest paths distance release.}
\algolab{private_path}
\end{algorithm2e}

As summarized in \algoref{private_path}, for a general weighted graph $G(\mathbb V,\mathbb E,w)$, our algorithm proceeds as follows:
\begin{enumerate}[1)]

\item Sample $n^{1/2}$ vertices from $\mathbb V$ uniformly to form set $\mathbb V_1$;

\item Create an edge set  $\mathbb{E}_1:=\{(u,v): u\in \mathbb V_1,v\in \mathbb V_1,u\neq v\}$.
For each $e=(u,v)\in \mathbb{E}_1 $,  set $w(e)=d(u,v)$ as the true shortest path distance;
	
\item Add $Lap(\mu_1,\sigma_1)$  noise to each edge $e\in \mathbb{E}_1$, i.e., $w'=w+Lap(\mu_1,\sigma_1)$;

\item Add $Lap(\mu_0,\sigma_0)$  noise to each edge $e\in \mathbb{E}_0$, i.e., $w'=w+Lap(\mu_0,\sigma_0)$ where
$\mathbb{E}_0:= \mathbb E\setminus \mathbb{E}_1$;

\item Obtain the merged graph $G'=(\mathbb V,\mathbb{E}'=\mathbb{E}_0 \cup \mathbb{E}_1,w')$, and compute all pairwise distances on $G'$.
\end{enumerate}

\noindent\textbf{Comments.} In \algoref{private_path}, the final estimated distances are calculated by using standard APSP distance computation on graph $G'$. One implication of this construction is that, beyond outputting the pairwise distances with privacy, we can in fact also publish the graph $G'(\mathbb V,\mathbb E',w')$ privately and allow the users to compute the private distances using standard algorithms by themselves. This is due to the positive mean of the Laplace noises. If we instead add zero-mean (centered) Laplace noises as in most standard approaches, then computing the APSP distances on $G'$ would not ensure the desired approximation error rate. We will provide more discussion on the role of the shifted noises at the end of this section.

Next, we provide theoretical analysis of our proposed algorithm. Firstly, recall that $\mathbb V_1$ contains $n^{1/2}$ uniformly sampled nodes from $\mathbb V$. Our analysis starts with the following fact.

\vspace{0.15in}

\begin{lemma}[Sampling Intersection] \lemlab{lem:sample}
Suppose $U\subseteq \mathbb V$ is a fixed vertex set and $|U|=n^{1/2}\log (n^2/\gamma)$. With probability at least $1-\gamma/n^2$, we have that $U\cap \mathbb V_1 \neq \emptyset$.
\end{lemma}
\begin{proof}
Suppose we pick a random vertex $v$ from $\mathbb V$ each time, the probability that $v$ is not in $\mathbb V_1$ is $1-1/n^{1/2}$. Then the probability of interest is bounded by $Pr[U\cap \mathbb V_1 = \emptyset] \leq (1-1/n^{1/2})^{(\sqrt{n} \log (n^2/\gamma))}=\gamma/n^2$.
\end{proof}

\vspace{0.15in}

\subsection{Privacy Analysis}

We show that \algoref{private_path} indeed achieves $(\epsilon,\delta)$-DP.

\begin{lemma}
\algoref{private_path} achieves $(\epsilon,\delta)$-DP.
\end{lemma}
\begin{proof}
The privacy budget is divided into two parts:

\vspace{0.1in}
(Part 1) The noise added to $\mathbb{E}_0$. It is obvious that for two neighboring inputs differ in the total weights of edges in $\mathbb{E}_0$ by $1$, by adding Laplace noises according  $Lap(\mu_0,1/\epsilon')$, it achieve $(\epsilon',0)$-DP. Note that adding a constant $\mu_0$ independent of the edge weights to the edges does not affect weight privacy.

\vspace{0.1in}

(Part 2) The noise added to edges in $\mathbb E_1$. Similar arguments hold. There are at most $n^{1/2}\cdot n^{1/2}=n$ pairs in $\mathbb E_1$. Applying Laplace mechanism with composition theorem (\lemref{lem:ACT}), we know that adding noise following $Lap(\mu_1,\sigma_1)$ with $\sigma_1=\sqrt{8n\log (1/\delta})/\epsilon'=2\sqrt{2} n^{1/2}\sqrt{\log (1/\delta)}/\epsilon'$ suffices to achieve $(\epsilon',\delta)$-DP. Applying simple DP composition theorem again proves the $(2\epsilon',\delta)$-DP. We conclude by noticing that $\epsilon'=\epsilon/2$.
\end{proof}

\newpage

\subsection{Sublinear Bound on the Approximation Error}

Firstly, we have the following fact that with high probability, any long path in $G$ would contain at least two nodes in $\mathbb V_1$ with high probability.

\begin{lemma}\lemlab{lem:two_intersect}
For a given  path in $G(\mathbb V,\mathbb E,w)$  with number of edges  larger than  $2n^{1/2}\log (n^2/\gamma)$, it  contains at least two vertices  in $\mathbb V_1$ with probability $1-2\gamma/n^2$.
\end{lemma}
\begin{proof}
Let us divide the path into three parts, such that two of them have at least $n^{1/2}\log (n^2/\gamma)$ edges.
Based on \lemref{lem:sample}, we have that each part above intersects with $\mathbb V_1$ with at least one vertex, each with probability at least $1=\gamma/n^2$. As a result, by union bound, the whole path must contain at least two nodes in $\mathbb V_1$ with probability $1-2\gamma/n^2$.
\end{proof}

Recall $G'=(\mathbb{V},\mathbb{E}'=\mathbb{E}_0 \cup \mathbb{E}_1 ,w')$ is the constructed synthetic graph. Further define $G'_w=(\mathbb V,\mathbb E',w)$ as the graph with nodes $\mathbb V$, edges $\mathbb E'$ and weights $w$ (notice that it is different from $G'=(\mathbb V,\mathbb E',w')$), and now consider those canonical shortest paths found by the true weights $w$ on $G'_w(\mathbb V,\mathbb E',w)$. We show that each canonical shortest path $P^{\mathbb V_1}_{u,v}$, $u,v\in\mathbb V$ in $G'_w$ has error $|d(P_{u,v},w')-d(P_{u,v},w)|$ bounded by $\tilde O(n^{1/2})$, where $d(P_{u,v},w)$ is defined by \defref{def:shortest path}.

\begin{lemma} \lemlab{lem:canonical}
For any $u,v\in\mathbb V$, let $P^{\mathbb V_1}_{u,v}$ be the canonical shortest path (\defref{def:canonical}) found by the true weights $w$. Then, it holds that $|d(P^{\mathbb V_1}_{u,v},w')-d(P^{\mathbb V_1}_{u,v},w)|=O(n^{1/2} \sqrt{\log(1/\delta)}\log^2 (n/\gamma)/\epsilon) $ with probability $1-4\gamma$ for $\forall u,v \in  \mathbb V$.
\end{lemma}

\begin{proof}
Firstly, we have $|d(P^{\mathbb V_1}_{u,v},w')-d(P^{\mathbb V_1}_{u,v},w)|=|\sum_{e\in P^{\mathbb V_1}_{u,v}} (w'(e)-w(e))|$, and $w'(e)-w(e)$ is the corresponding Laplace noise added to edge $e$ in $P^{\mathbb V_1}_{u,v}$. The key observation in the proof is that, each canonical shortest path $P^{\mathbb V_1}_{u,v}$ can be divided into three parts $(u,...,p),(p,q), (q,...,v)$, where $p,q$ are the closest nodes in $\mathbb V_1$ to $u,v$, respectively, and the edge $(p,q)$ is the constructed shortcut provided by $\mathbb V_1$ and $\mathbb E_1$. Note that, $P^{\mathbb V_1}_{u,v}$ also might not contain a shortcut, but this does not matter in our analysis in the sequel. Denote $P_{u,v,}$ as the shortest path between $u$ and $v$. Consider two cases:
\begin{itemize}
    \item $|P_{u,v}|\leq 2n^{1/2}\log(n^2/\gamma)$. In this case, we know that the total length of $(u,...,p)$ and $(q,...,v)$ is bounded by $2n^{1/2}\log(n^2/\gamma)$. That is, $P^{\mathbb V_1}_{u,v}$ contains at most $2n^{1/2}\log(n^2/\gamma)$ edges from $\mathbb E_0$. By Laplace mechanism, with probability $1-\gamma$, every edge in $\mathbb E_0$ leads to at most $O(\log(n/\gamma)/\epsilon)$ error. Thus, edges in $\mathbb E_0$ contribute at most $O(n^{1/2}\log^2(n/\gamma)/\epsilon)$ error. $P^{\mathbb V_1}_{u,v}$ also contains at most one shortcut edges from $\mathbb V_1$, which gives $O(n^{1/2}\log(n/\gamma)\sqrt{\log(1/\delta)}/\epsilon)$ error, with another probability $1-\gamma$. Thus, with probability $1-2\gamma$, the error is bounded by $O(n^{1/2}\log^2(n/\gamma)/\epsilon)$.

    \item $|P_{u,v}|> 2n^{1/2}\log(n^2/\gamma)$. In this case, by \lemref{lem:two_intersect} and union bound, we know that with probability $1-\frac{2\gamma}{n^2}n^2=1-2\gamma$, $P^{\mathbb V_1}_{u,v}$ contains a shortcut $(p,q)$, and both the lengths of $(u,...,p)$ and $(u,...,q)$ are upper bounded by $n^{1/2}\log(n^2/\gamma)$. The remaining proof is similar to the arguments above. We have that with probability at least $1-4\gamma$, $|d(P^{\mathbb V_1}_{u,v},w')-d(P^{\mathbb V_1}_{u,v},w)|$ is bounded by $O(n^{1/2}\log^2(n/\gamma)/\epsilon)$.
\end{itemize}

In summary, we have shown that with probability at least $1-4\gamma$, $|d(P^{\mathbb V_1}_{u,v},w')-d(P^{\mathbb V_1}_{u,v},w)|$ is bounded by $O(n^{1/2}\log^2(n/\gamma)/\epsilon)$. The proof is complete.
\end{proof}

Additionally, we have the following result stating that the estimated distance is no smaller than the true distance for all pairs of vertices, with high probability.

\begin{lemma} \lemlab{lem:gamma}
Let $P_{u,v}$ be the shortest path between $u,v\in\mathbb V$ on $G$, and $P'_{u,v}$ the shortest path found on $G'$. Then with probability $1-2\gamma$, $d(P'_{u,v},w')\geq d(P_{u,v},w)$ for all $u,v \in \mathbb{V}$.
\end{lemma}
\begin{proof}
By adding each edge in $\mathbb{E}_1$ with noise according to $Lap(\mu_1,\sigma_1)$  where $\sigma_1=2\sqrt{2} n^{1/2}\sqrt{\log (1/\delta)}/\epsilon'$ and $\mu_1=\sigma_1 \log (n/\gamma)$, with high probability $1-\gamma$, every $Lap(\mu_1,\sigma_1)$ variable is greater or equal to $\mu_1- \sigma_1 \log (n/\gamma)\geq 0$ with probability $1-\gamma$ based on \lemref{lem:single}.

Similarly, for edges in $\mathbb{E}_0$, we add to each of them $Lap(\mu_0,\sigma_0)$ noise where $\sigma_0=1/\epsilon',\mu_0=\sigma_0 \log (n^2/\gamma)$.
Then  with another high probability $1-\gamma$, every variable according to $Lap(\mu_0,\sigma_0)$ is non-negative based on \lemref{lem:single}.
Therefore, since with high probability the added noises on all the edges are positive, the noisy shortest path distance must be larger than the true distance, i.e., $d(P'_{u,v},w')$. This proves the claim.
\end{proof}

Now, we are in the position to present the main error guarantee of this work.

\begin{theorem} \thmlab{main}
Let $G = (\mathbb{V}, \mathbb{E},w)$ be general graph with $n$ vertices. For some $\epsilon,\delta>0$, running \algoref{private_path} publishes the graph $G'$ with $(\epsilon,\delta)$-differential privacy w.r.t. the weights $w$. Let $\hat d(u,v)$ be the estimated shortest path distance between $u,v\in\mathbb V$ computed on $G'$. Then, with probability $1-4\gamma$, we have
$|\hat d(u,v)-d(u,v)|\leq O(\epsilon^{-1}n^{1/2}\log^2(n/\gamma)\sqrt{\log (1/\delta)})$ for all $u,v\in\mathbb V$.
\end{theorem}
\begin{proof}
Let $P'_{u,v}$ be the shortest path between $u$ and $v$ in the synthetic graph $G'(\mathbb V, \mathbb E', w')$, then trivially, $d(P'_{u,v},w') \leq d(P^{\mathbb V_1}_{u,v},w')$. By \lemref{lem:canonical}, we know that $|d(P^{\mathbb V_1}_{u,v},w')-d(P^{\mathbb V_1}_{u,v},w)|$ is bounded by $ O(n^{1/2} \sqrt{\log(1/\delta)}\log^2(n/\gamma)/\epsilon)$ with probability $1-4\gamma$, $\forall u,v \in  \mathbb V$. In this event, we have $d(P'_{u,v},w') \geq  d(P_{u,v},w)$ based on \lemref{lem:gamma}. Notice that, $d(P_{u,v},w)\equiv d(P^{\mathbb V_1}_{u,v},w)$ by definition. Therefore, with  probability $1-4\gamma$,
\begin{align}
    |\hat d(u,v)-d(u,v)|=|d(P'_{u,v},w')-d(P_{u,v},w)|&\underset{(a)}{\leq} |d(P^{\mathbb V_1}_{u,v},w')-d(P^{\mathbb V_1}_{u,v},w)|  \nonumber\\
    &\leq O(n^{1/2} \sqrt{\log(1/\delta)}\log^2(n/\gamma)/\epsilon), \label{eq1}
\end{align}
which verifies the claim.
\end{proof}

\noindent\textbf{Comments.} As we see, the positive mean of Laplace noises implies \lemref{lem:gamma} which provides a key inequality in the proof of \thmref{main}. This result in fact allows us to only consider the error on the positive side. On the contrary, if we use mean-zero Laplace noises, then we can no longer bound $|d(P'_{u,v},w')-d(u,v)|$ by $|d(P^{\mathbb V_1}_{u,v},w')-d(u,v)|$ as in \lemref{lem:canonical}, since it is possible that $d(P'_{u,v},w')\leq d(P^{\mathbb V_1}_{u,v},w')<d(u,v)$. Therefore, if we simply add zero-mean noises, then directly computing the APSP distances on the constructed graph $G'$ would not guarantee the desired error level.

\begin{figure}[b!]
\mbox{
	\includegraphics[width=3.25in]{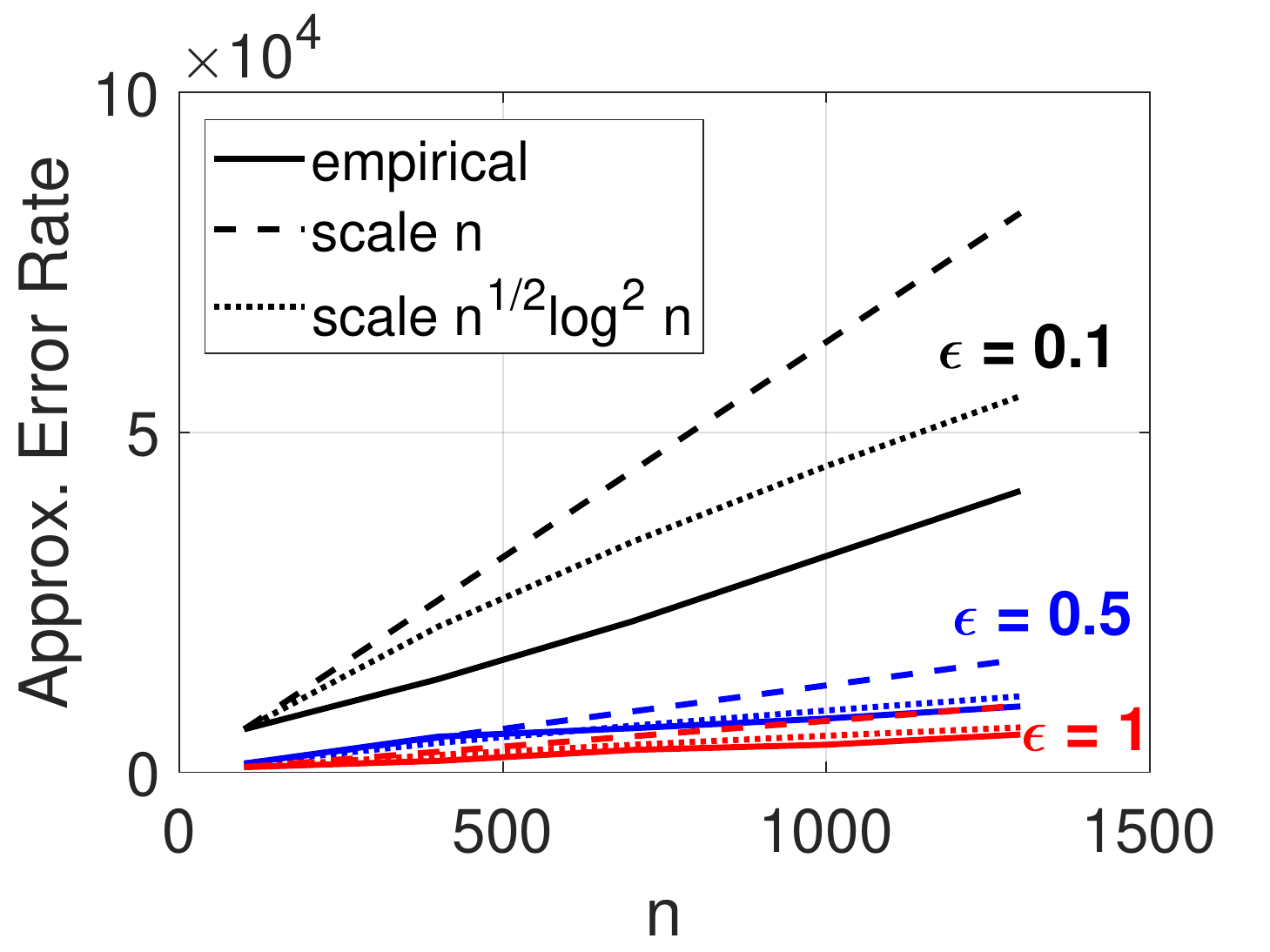}
	\includegraphics[width=3.25in]{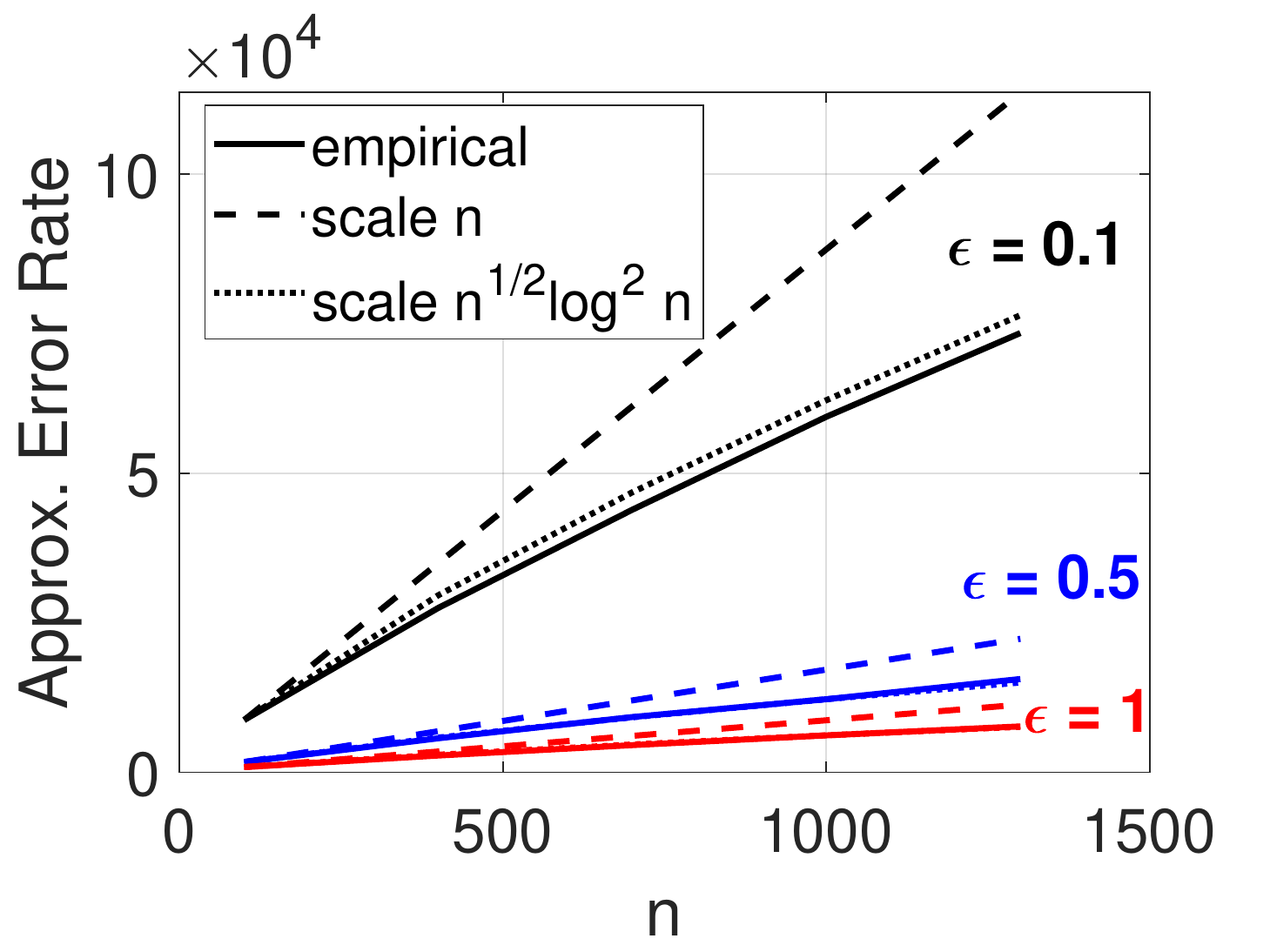}
}

\vspace{-0.1in}

	\caption{Empirical approximation error rate vs. graph size $n$ (solid curve), with predicted growth rate of $O(n)$ and $O(n^{1/2}\log^2 n)$ (dash and dotted curves). $\delta=0.01$, $\gamma=0.01$. Left: each edge weight is $Unif(2000,3000)$. Right: each edge weight follows $Unif(1e4,1e5)$. Results are averaged over 200 repetitions.}\figlab{error}
\end{figure}

\subsection{Numerical Example}	We present a numerical example on a multi-stage graph to justify the theory. We simulate a multi-stage graph, where each edge weight is generated from i.i.d. uniform distribution. Each stage contains 1 start node, 1 end node, and 9 nodes each connected to both the start and the end node. We consider multi-stage graphs here because in some sense it is one of the ``hardest'' cases for private APSP computation, since many shortest paths contain $O(n)$ edges. In \figref{error}, we present the empirical largest absolute error over all pairwise distances against the size of graph $n$. For the theoretical bounds, we start with the first empirical error (black) associated with $n=100$, and increase at the rate of $n$ (dash) and $n^{1/2}\log^2 n$ (dotted), respectively. We validate that the empirical error grows with rate slower than $n^{1/2}\log^2 n$, as given by our theory, for all $\epsilon$ values. Note that the error bound becomes tighter with large edge weights (right) than with small weights (left). This can be explained by (a) in Eq. (\ref{eq1}). When the edge weights follow $Unif(1e4,1e5)$ which is much larger than the noise magnitude, the shortest path $P'_{u,v}$ found on $G'$ would be the same as $P_{u,v}$ (or $P^{\mathbb V_1}_{u,v}$) in most cases. Thus, the LHS and the RHS of (a) would be close, making the bound tighter.

\vspace{0.05in}

\section{ Graph with Small Feedback Vertex Set}

\begin{algorithm2e}[b!]
\SetKwInput{Input}{Input}
\SetKwInOut{Output}{Output}
\Input{General graph $G=(\mathbb{V},\mathbb{E},w)$, private  parameter $\epsilon,\delta$ .}
\DontPrintSemicolon
$\epsilon'=\epsilon/3$.\;

Compute a feedback vertex set $S$ of $G$ using the 2-approximation algorithm in~\citet{DBLP:journals/ai/BeckerG96}. \hfill {\color{blue}// Step 1)}\;

Let $G'$ be the induced subgraph of vertex set $\mathbb V\setminus  S$ in $G$.\;

Compute the differentially private  all pairwise shortest path distances $\hat d(u,v)$  on trees/forest $G'$ using Algorithm 1 in~\cite{DBLP:conf/pods/Sealfon16} with privacy parameter $\epsilon'$. \hfill {\color{blue}// Step 2)}\;

Let $k:=|S|$.\;
Compute  $d(u,v)$ for $u, v\in S$, where $d(u,v)$  is the shortest path distance between $u$ and $v$. \hfill {\color{blue}// Step 3)}\;

\For {each   $u\in S,v\in S$ }
{
Let 	$\sigma_1:= (2\sqrt{2} k\sqrt{\log (1/\delta)})/\epsilon'$.\;
	Compute $X_e:= Lap(0,\sigma_1)$.\;
	$\hat d(u,v):= d(u,v)+X_e$.\;
}

\For {each vertex $u\in  S$}
{
 \For {each vertex $v \notin S$ adjacent to $u$} \hfill {\color{blue}// Step 4): To deal with the edges between $S$ and $\mathbb V \setminus S$\; }
 {
	Let $\sigma_0:= 1/\epsilon'$.\;
	Compute $X_e:= Lap(0,\sigma_0)$.\;
	$w'(u,v):= w(u,v)+X_e$.\;
	$\hat d(u,v):=w'(u,v)$.
 }
}

\For {each vertex $u$ in $\mathbb V \setminus S$}
{
 \For {each vertex  $v$ in $S$}
 {
 {\color{blue}// Step 4): To deal with the case that only the last vertex of shortest path is in $S$\; }

    $\hat d(u,v):=\min_{p \notin S}\{\hat d(u,p)+ w'(p,v) \}$.  \;

 }
}

\For {each vertex $u$ in $\mathbb V \setminus S$}
{
 \For {each vertex  $v$ in $S$}
 {
 {\color{blue}// Step 5): To deal with the case that shortest path  passes a vertex $p$ in $S$\; }

    $\hat d(u,v):=\min_{p \in S}\{\hat d(u,p)+ \hat d(p,v) \}$.\;

 }
}

\For {each pair  $u,v \in \mathbb V \setminus S$}
{

 {\color{blue}// Step 6): To deal with the case that shortest path  passes some vertex $u$ in $S$\; }

 $\hat d(u,v):=\min_{p \in S}\{\hat d(u,p)+ \hat d(p,v) \}$.\;

}

\Output{All pairwise distances $\hat d(u,v)$.}
\caption{Differentially private all pairwise shortest path distances release for graph with small feedback vertex number.}
\algolab{private_graph_feedback}
\end{algorithm2e}

In previous section, we have investigated the problem of DP all pairwise shortest distance problem on general weighted graphs. In this section, we consider a more special type of graph with small feedback vertex set (FVS), which is a common concept in graph analysis, e.g.,~\citep{Article:Stefan2010,Article:Charis2020}. We will design a new algorithm based on FVS computation that privately releases all pairwise distances. We will show that this method improves the error on general graphs when the FVS number is small. To begin with, the definition of FVS is formally stated as below.

\newpage

\begin{definition}[Feedback Vertex Set~\citep{DBLP:journals/ai/BeckerG96}]
Let $G=(\mathbb V, \mathbb E)$ be an undirected graph. A set $X \subseteq \mathbb V$ is called a feedback vertex set (FVS) if $G\setminus X$ is a forest, where $G\setminus X$ is the induced graph by $\mathbb V\setminus X$. The feedback vertex set number is the minimal cardinality over all the possible feedback sets.

\end{definition}

In words, a feedback vertex set (FVS) of a graph is a subset of vertices $X$ such that after removing these nodes from the graph, the subgraph induced by $\mathbb V\setminus X$ contains no cycles (i.e., is a forest). The FVS number is the smallest size of all the FVS's.

Next, we propose a differentially private algorithm that releases the APSP distances of graphs based on the feedback vertex set computation. The steps of \algoref{private_graph_feedback} are summarized below:
\begin{enumerate}[1)]
    \item We  compute a feedback vertex set $S$ of $G$ by using the 2-approximation algorithm in~\citet{DBLP:journals/ai/BeckerG96}, with $|S|=k$. Then the induced subgraph $G'$  of  vertex set $\mathbb V \setminus S$ is a forest;

    \item We use Algorithm 1 in~\citet{DBLP:conf/pods/Sealfon16} to obtain the private all pairwise shortest path distances in $G'$. For any pair $(u,v)$ in $G$ whose shortest path does not pass any vertex in $S$, $\hat d(u,v)$ on $G'$ is already a good estimation for $d(u,v)$;

    \item For the $k^2$ (recall $k=|S|$) pairwise distances of the node pairs in $S$, we can compute the APSP distances directly and add $Lap(0,\sigma_1)$ to achieve DP. We still use the $\hat d(u,v)$ to represent them;

    \item For each edge $(u,v)$ with  $u\in  S,v\in \mathbb V \setminus S$, we add noise according $Lap(0,\sigma_0=1/\epsilon')$ to $w(u,v)$.
    For a pair in $\{(u,v)|u\in \mathbb V\setminus S, v\in S\}$, if the shortest path between $u$ and $v$ does not pass other vertex in $S$ except $v$, then  there exist some neighbor $p$ of $v$ such that $d_{G}(u,v)=d_{G'}(u,p)+w(p,v)$ based on ~\lemref{lem:left}, where $d_G(u,v),d_{G'}(u,v)$ represent the shortest path distance between $u$ and $v$ in $G$ and $G'$ respectively. Hence we use  $\hat d(u,v):=\min_{p\notin S}\{\hat d(u,p)+ w'(p,v) \}$  to estimate $d(u,v)$;

    \item For a pair in $\{(u,v):u \in \mathbb V\setminus S, v\in S\}$, if the shortest path between $u$  and $v$ passes some  vertex $p \in S$ except $v$, then  there exist $p \in S, p'\notin S$ such that $d_{G}(u,v)=d_{G}(u,p)+d_{G}(p,v)$, and the shortest path between $u$ and $p$ does not pass any vertex in $S$ except $p$ based on \lemref{lem:right}. We obtained $\hat{d}(u,p)$ in step 4). We then use  $\hat d(u,v):=\min_{p \in S}\{\hat d(u,p)+ \hat d(p,v) \}$  to estimate $ d(u,v)$;

    \item For pairs in $\{(u,v) \in (\mathbb V \setminus S )^2\}$, if the shortest path between $u$ and $v$ passes  some vertex $p$ in $S$, then $d_G(u,v)= d_G(u,p)+ d_G(p,v)$. We can use  $\hat d(u,v):=\min_{p \in S}\{\hat d(u,p)+\hat d(p,v)\}$  to estimate $d(u,v)$, also  $\hat d(u,p),\hat d(p,v)$ had been computed in 4), 5).

\end{enumerate}

\vspace{0.1in}

The two auxiliary lemmas mentioned in step 4) and 5) are given below.

\begin{lemma}\lemlab{lem:left}
For a given graph $G(\mathbb V,\mathbb E,w)$, let  $G'$ be the induced subgraph of $G$ with vertex set $\mathbb V \setminus S$. For any pair $\{(u,v)|u\in \mathbb V\setminus S, v\in S\}$, if the shortest path between $u$ and $v$ does not pass other nodes in $S$, then there exist a  neighbor $p$ of $v$ such that $d_{G}(u,v)=d_{G'}(u,p)+w(p,v)$.
\end{lemma}
\begin{proof}
Note that the shortest path between $u$ and $v$ does not pass other nodes in $S$ except $v$. Let the second last vertex  from $u$ to $v$ be $p$, with $p\notin S$.
Then we have $d_{G}(u,p)=d_{G'}(u,p)$, thus $d_{G}(u,v)=d_{G}(u,p)+w(p,v)=d_{G'}(u,p)+w(p,v)$.
\end{proof}

\vspace{0.1in}

\begin{lemma}\lemlab{lem:right}
Let graphs $G$ and $G'$ be defined as in \lemref{lem:left}. For any pair $\{(u,v)|u\in \mathbb V\setminus S, v\in S\}$, assume that the shortest path  $P_{u,v}$ between $u$ and $v$ does pass other nodes in $S$ except $v$. Then there exist some  vertex $p \in S, p'\notin S$ such that $d_{G}(u,v)=d_{G}(u,p)+d_{G}(p,v)$ and $d_{G}(u,p)=d_{G'}(u, p')+w(p',p)$.
\end{lemma}
\begin{proof}
Let $p$ be the first vertex  in $S$ that  $P_{u,v}$ passes through, and denote $p'$ as the vertex just before $p$ in path $P_{u,v}$. We have that $d_G(u,p')=d_{G'}(u,p')$. Then we know that $d_{G}(u,v)=d_{G}(u,p)+d_{G}(p,v)$ and $d_{G}(u,p)=d_{G'}(u, p')+w(p',p)$.
\end{proof}

\newpage

\subsection{Privacy Analysis}

\begin{lemma}
\algoref{private_graph_feedback} achieves $(\epsilon,\delta)$-DP.
\end{lemma}
\begin{proof}
The privacy budget $\epsilon$ is divided into three parts:

\vspace{0.1in}
(Part 1) For the distances between all pairs in $\mathbb V \setminus S$, our method achieves $(\epsilon',0)$-DP by the result on trees from~\citep{DBLP:conf/pods/Sealfon16}.

\vspace{0.1in}

(Part 2) For the $k^2$ distances of all node pairs in $S$, by adding to each edge weights i.i.d. $Lap(0,\sigma_1)$ noises with $\sigma_1=\sqrt{8n\log (1/\delta})/\epsilon'=2\sqrt{2} n^{1/2}\sqrt{\log (1/\delta)}/\epsilon'$, we can achieve $(\epsilon',\delta)$-DP according to \lemref{lem:ACT}.

\vspace{0.1in}

(Part 3) For each edge $(u,v)$ with  $u \in \mathbb V\setminus S, v\in S$, obviously by Laplace mechanism $w'(u,v):=w(u,v)+Lap(1/\epsilon')$, we achieve $(\epsilon',0)$-DP to release all the pairwise distances.

Composing three privacy budges up and using simple composition theorem of DP, we show that \algoref{private_graph_feedback} achieves $(\epsilon,\delta)$-DP.
\end{proof}

\subsection{Error Bound}

We will make use of the following result regarding our step 2) on trees.

\begin{lemma}[All pairwise distance on trees~\citep{DBLP:conf/pods/Sealfon16}] \lemlab{lem:tree}
For a tree $T$ with non-negative edge weights $w$ and $\epsilon > 0$, there is an $\epsilon$-differentially private algorithm that releases
APSP distances such that with probability $1-\gamma$, all released distances have approximation error bounded $O(\log^{2.5}n \log(1/\gamma)/\epsilon)$.
\end{lemma}

\vspace{0.1in}
The error bound of \algoref{private_graph_feedback} is given as below.

\begin{theorem}  \thmlab{theo:small_feedback_error}
Let $G = (\mathbb{V}, \mathbb{E},w)$ be a general graph with $n$ vertices. For some $\epsilon,\delta>0$, running \algoref{private_graph_feedback} publishes all APSP distances that is $(\epsilon,\delta)$-differentially private w.r.t. the weights $w$. With probability $1-3\gamma$, the additive error is bounded by
$O(\epsilon^{-1}k\log(k/\gamma)\sqrt{\log (1/\delta)}+\epsilon^{-1}(\log^{2.5}n)\log(1/\gamma)) $.
\end{theorem}
\begin{proof}
For $u\in S, v\in S$, by adding  $d(u,v)$ with $Lap(0,\sigma_1)$ noise  where $\sigma_1=2\sqrt{2} n^{1/2}\sqrt{\log (1/\delta)}/\epsilon'$, with probability $1-\gamma$, $|\hat d(u,v)-d(u,v)|=O(k \sqrt{\log (1/\delta)}\log (k/\gamma)$ $\forall u,v \in S$.
Based on \lemref{lem:tree}, we have that with probability $1 - \gamma$, all released distances in $G'$ have approximation error $O(\log^{2.5}n\log(1/\gamma)/\epsilon) $.
For those edges $(u,v)$ with $u\in \mathbb V\setminus S, v\in S$, we add to each $w(u,v)$ a Laplace noise according to $Lap(0,\sigma_0)$ with $\sigma_0=1/\epsilon'$. Thus, with another  probability $1-\gamma$, $|w'(u,v)-w(u,v)| \leq O(\log (n/\gamma)/\epsilon)$ $\forall \{ u\in \mathbb V \setminus S, v\in S\}$.
Union bound implies that, with probability $1-3\gamma$, the total error is bounded by $O(\epsilon^{-1}k\log(k/\gamma)\sqrt{\log (1/\delta)})+O(\epsilon^{-1}(\log^{2.5}n)\log(1/\gamma))$ as claimed.

\end{proof}

\noindent\textbf{Comments.} We see that when $k=|S|=o(n^{1/2})$, the error in \thmref{theo:small_feedback_error} is smaller than the $O(n^{1/2}\log^2 n)$ error in \thmref{main} for the general weighted graphs. Therefore, our result implies improved error bounds when the graph has small feedback vertex set such that $k=o(n^{1/2})$. Further, if $k=o(\log n)$, i.e., the graph is ``almost'' a forest, then the first term vanishes and the error reduces to $\tilde O(\log^{2.5} n)$ for trees~\citep{DBLP:conf/pods/Sealfon16}.

\vspace{0.1in}
\section{Conclusion}

In literature, to achieve  $(\epsilon,\delta)$-differential privacy (DP) in general weight private graphs, releasing all pairwise shortest path (APSP) distances incurs maximal approximation error $\tilde O(n)$, linear in the number of vertices. Recently, \cite{fan2022distances} studied this problem on grid graphs and trees, but the linear error barrier for general weighted graphs is still an open problem~\citep{DPorg-open-problem-all-pairs}. In this work, we propose a differentially private algorithm to release a carefully constructed graph, and computing  the APSP distance on the synthetic graph achieves $\tilde O(n^{1/2})$ error sublinear in $n$, thus answering the open question. With improved and extended techniques from~\cite{fan2022distances}, the idea of our approach is to augment the diameter of graph by a set of shortcuts along with adding shifted Laplace noise. Moreover, for graphs with small feedback vertex set number $k$, we also propose a DP algorithm to answer all pairwise shortest path distances each with error $\tilde O(k)$. This improves the result for general graphs when $k$, the feedback vertex set number, is small.

\newpage

\bibliography{dp}
\bibliographystyle{plainnat}

\end{document}